\documentclass[runningheads,envcountsame,a4paper]{llncs}
\pdfoutput=1

\usepackage[utf8]{inputenc}
\usepackage{microtype}
\usepackage{amsmath}
\usepackage{amssymb}
\usepackage{float}
\usepackage[algoruled,linesnumbered,noend]{algorithm2e}
\usepackage{relsize}

\SetCommentSty{mycommentsty}
\usepackage{enumitem}
\usepackage{cite}
\usepackage{multirow}
\usepackage{url}
\usepackage{xspace}
\usepackage{graphicx}
\usepackage{booktabs}
\usepackage{mathcomp}
\usepackage{ifdraft}

\usepackage{verbatim}

\usepackage{tikz}
\usetikzlibrary{shapes,arrows}
\usetikzlibrary{snakes}
\usetikzlibrary{positioning,patterns}

\usepackage{ifpdf}\usepackage{datetime}
\ifpdf%
\fi
\usepackage{graphicx}
\usepackage{a4wide}
\newcommand{\grb}{\ensuremath{G_{RB}}\xspace}
\newcommand{\gc}{\ensuremath{G_{c}}\xspace}
\newcommand{\Cc}{\ensuremath{\mathcal{C}(c)}\xspace}

\newcommand{\grbd}{$G_{RB}$. }

\newcommand{\me}{\ensuremath{M_{e}}\xspace}

\newcommand{\ie}{i.e.,~}

\newcommand{\pp}{\textbf{Decide-pp}}

\begin{document}

\title{Algorithms for the  Constrained  Perfect Phylogeny with Persistent Characters}
\author{Paola Bonizzoni\inst{1} \and Anna Paola Carrieri \inst{1}  \and Gianluca Della Vedova \inst{1}  \and Gabriella Trucco \inst{2}
}

\institute{Dipartimento di Informatica Sistemistica e Comunicazione \\
Univ.  degli Studi di Milano - Bicocca \\
%Viale Sarca 336, 20126 Milano - Italy \\
\email{bonizzoni,annapaola.carrieri,dellavedova@disco.unimib.it}
\and
 Dipartimento di Tecnologie dell'Informazione
Univ. degli Studi di Milano, Crema \\
\email{gabriella.trucco@unimi.it} }

\pagestyle{plain}
\date{march}\maketitle

\begin{abstract}
The   perfect phylogeny is one of the most used models in different areas of computational biology.
In this paper we consider the problem of  the Persistent Perfect Phylogeny
(referred as P-PP) recently introduced to extend the perfect phylogeny  model  allowing {\em persistent} characters,
that is characters  can be  gained and lost at most once.
We define a  natural generalization of the P-PP problem obtained by requiring
that for some pairs (character, species), neither the species nor any of its
ancestors can have the character.
In other words, some characters cannot be persistent for some species.
This new problem is called {\bf Constrained P-PP} (CP-PP).
  Based on a graph formulation of the CP-PP problem, we are able to provide a polynomial time solution   for the CP-PP problem for matrices having an empty conflict-graph.
  In particular we show that all such matrices admit a persistent  perfect
  phylogeny in the unconstrained case.
Using this result,   we develop  a  parameterized algorithm for solving  the CP-PP problem where the parameter is the number of characters.
A preliminary  experimental analysis of the  algorithm  shows that it performs efficiently  and it may analyze  real haplotype  data
not conforming to  the classical perfect phylogeny model.
%  persistent model explain it can detect persistent it is applicable  to incorporate biological complexity due to  back %mutations events, i.e. gain and loss of characters.
 \end{abstract}

\vspace{-.2in}
\section{Introduction}

The perfect phylogeny is one of the most investigated and used models  in several applications where a coalescent model is required for analyzing genomic data.
Conceptually the  model is based on the infinite sites assumption, that is no character can mutate more than once in the whole tree.
While this assumption is quite restrictive, the perfect phylogeny model
turned out to be splendidly coherent within the haplotyping problem~\cite{bo4,Gus02}, where we want to distinguish the two
haplotypes present in each individual when given only genotype data.
More precisely, the interest here is in computing a set of haplotypes and a perfect phylogeny such that the haplotypes
(i) label the vertices of the perfect phylogeny and (ii) explain the input set of genotypes.
This context has been deeply studied in the last decade, giving rise to a number of   algorithms~\cite{Boniz,Gus06}.

In the general framework,  the model is used to reconstruct the evolution of species (taxa) characterized by a set of  binary characters that are gained and/or lost during the evolution. Characters can take only the states $0$ or $1$, usually interpreted as the presence or absence of the character. Restrictions on the type of changes from zero to one and vice versa lead to a variety of specific models \cite{Fel}.

Still, the perfect phylogeny model and the assumptions that have been central in the previous decades cannot be employed
without adaptations or improvements.
In fact, this model is too restrictive to explain the biological complexity of
real data, where homoplasy events (such as recurrent mutations or back
mutations) are present. Thus a central goal in this model is to extend its
applicability, while retaining  the computational efficiency.
 The problem of constructing  phylogenies where the deviations from perfect phylogeny are  small
  has been tackled  for example in \cite{DBLP:journals/siamcomp/Fernandez-BacaL03} under the name of near perfect phylogeny. In particular,  the near perfect phylogeny  haplotyping problem
has been explored in   \cite{DBLP:conf/ismb/SatyaMAPB06}.
Especially the impossibility of losing a character that has been previously acquired turned out to be too restrictive, resulting in more elaborated models, such as the notion of persistent character ~\cite{Pr1} and the General Character Compatibility~\cite{DBLP:conf/isbra/ManuchPG11}.

Following the research direction proposed in ~\cite{Pr1} to investigate the
dynamics of  protein interactions,  the Persistent Perfect Phylogeny problem has
been introduced~\cite{DBLP:journals/tcs/BonizzoniBDT12} to
address the computational  problem of constructing a perfect phylogeny under the
assumption that  only a special type of back mutation may occur in the tree: a
character may change state at most twice in
the tree: once from $0$ to $1$, and (maybe) once from $1$ to $0$.

In the paper we   consider a natural generalization of the P-PP problem, called Constrained P-PP problem (CP-PP), that is obtained by adding a constraint for some characters in the input data,  given by the fact they cannot be persistent in some species. Then we explore algorithmic solutions for the CP-PP problem as well as for the restricted case of the P-PP problem.

%Thus results achieved for the GCC  problem apply also to the CP-PP  and P-PP  problems.

Since our main aim is to find algorithmic solutions,  following the approach in
~\cite{DBLP:journals/tcs/BonizzoniBDT12}, we first explore a graph formulation
of the CP-PP problem, called {\bf Red-black graph reduction} based on the
equivalence of  P-PP to  a problem of completing a matrix where each character
$c$ has two columns $c^+, c^-$, with  $c^+$ ($c^+$) equal to $1$ in a species
$s$ in the matrix  corresponds  to the fact that $s$ has  gained (lost) the character $c$.
Based on the above graph formulation,  we prove that  there exists a   class of binary matrices that  always admit  a positive solution for the P-PP problem, that is they admit a \emph{persistent perfect phylogeny}. For this   special case we also provide a polynomial algorithm that works for the general CP-PP problem.  Based on this polynomial time algorithm  we  show that CP-PP is fixed parameter tractable in the number of characters and propose  a branch and bound based  algorithm,  called {\bf  \pp -opt}.

The algorithms and models discussed in the paper  may have interesting  applications in the  construction of evolutionary trees based on the analysis of binary genetic markers, where variants of the perfect phylogeny have already been considered, such as  in the study of evolution based on  introns~\cite{Pr1} or progression pathways using tumor markers or  in discovering significant associations between phenotypes and single-nucleotide polymorphism markers \cite{DBLP:conf/psb/PanMVTW09} and also in haplotype analysis.
In the paper we have run a  preliminary experimental analysis  aimed to show the applicability of our algorithm for the CP-PP model to  deal with  biological data (binary characters) incorporating recombination events.
The results show that the algorithm performs efficiently  on simulated matrices as well as on real data. The algorithm applied to haplotypes taken from the HapMap project is
able  to detect in binary matrices characters that may be persistent,  whenever they  do not obey the standard perfect phylogeny model.

Finally, we observe that the CP-PP problem  (P-PP problem) is equivalent to cases of the
 General Character Compatibility problem (GCC) investigated in ~\cite{DBLP:conf/isbra/ManuchPG11} whose complexity is still open. Thus our results also apply to  those problems.

\section{The persistent perfect phylogeny and the  red-black graph reduction problem}

In this section we present  the persistent perfect phylogeny problem and its
constrained version.
We show that the  CP-PP problem  reduces to a graph problem,  the  {\bf red-black graph reduction} problem: this result  generalizes the  Theorem \ref{main-equivalence} given in ~\cite{DBLP:journals/tcs/BonizzoniBDT12} to the  constrained persistent perfect phylogeny problem.

In this paper the input is an  $n \times m$ binary matrix $M$ whose columns are associated
with the set $C  = \{c_1, \ldots,
c_m\}$ of characters and whose rows are associated with the set $S = \{s_1, \dots,
s_n\}$ of species.
Then $M[i,j] = 1$ if and
only if the species $s_i$ has character $c_j$, otherwise $M[i,j] = 0$.
The gain of a character $c$ in a phylogenetic tree is represented by an
edge labeled by the character $c^{+}$. In order to model  the presence of
persistent characters, the loss of a character $c$ in the tree is represented
by  an edge that is labeled by the negated character, labeled by $c^{-}$.
Formally, we define the persistent perfect phylogeny model as follows~\cite{zeng,DBLP:journals/tcs/BonizzoniBDT12}.
% \vspace{.1in}
% \noindent

\begin{definition}[Persistent Perfect Phylogeny]
\label{def:persistent-perfect-phylogeny}
Let $M$ be an  $n \times m$ binary matrix.
Then a \emph{persistent perfect phylogeny}, in short  \emph{ p-pp},
for $M$ is a rooted tree $T$ such that:
\begin{enumerate}
\item each node $x$ of $T$ is labeled by a vector $l_x$ of length $m$;
\item
  the root of $T$ is labeled by a vector of all zeroes,  while  for each  node
$x$ of $T$ the value $l_x[j]\in\{0, 1\}$
  represents the state of character $c_j$ in tree $T$;
\item
  each edge $e=(v,w)$ is labeled by at least a character;
\item
  for each character $c_j$ there are at most  two  edges $e=(x,y)$ and
$e'=(u,v)$
  such that $l_x[j] \neq l_y[j]$ and $l_u[j] \neq l_v[j]$
  (representing a change in the state of $c_j$).
In that case $e$, $e'$  occur along the same path
  from the root of $T$ to a leaf of $T$; if $e$ is closer to the root than $e'$,
  then $l_x[j]=l_v[j]=0$, $l_y[j]=l_u[j]=1$, and
  the  edge $e$ is labeled $c_j^{+}$,
  while  $e'$ is labeled $c_j^{-}$;
\item
  each row $r$ of $M$ labels exactly one node $x$ of $T$.
  Moreover the vector $l_{x}$ is equal to the row $r$.
\end{enumerate}

%Then we say that $T$ \emph{realizes} matrix $M$, that is $M$ admits a p-pp tree.
%Whenever  a binary matrix $M$ admits a p-pp tree we say  that $M$ is \emph{solvable}.
\end{definition}

Let us state the main problems investigated in the paper.

{The \bf  Persistent Perfect Phylogeny
 problem (P-PP):} Given a binary matrix $M$, returns a p-pp tree for $M$ if such a tree exists.

A natural generalization of the P-PP problem is obtained by considering as input data a  pair $(M, E^*)$, called {\em constrained matrix} where $M$ is
a matrix over binary alphabet and  $E^*=\{(i_1,j_1), \cdots (i_k, j_k)\}$ is a set of  $0$-entries of $M$. Now a p-pp tree  $T$ for matrix $M$ is consistent with $E^*$
if and only if  $(i,j) \in E^*$ means that character $c_j$ is absent in species $ s_i$ and
  cannot be gained and then lost in the species  $s_i$ in tree $T$, that is $c_j$ is not persistent in species $s_i$.

 {The \bf  Constrained Persistent Perfect Phylogeny
 problem (P-PP):} Given a constrained  matrix $(M, E^*)$,  returns a p-pp tree $T$ for $M$ if such a tree exists,  such that   tree $T$  is consistent with $E^*$.

 Another fundamental observation is that we can restate the CP-PP  problem  as
variant  of the Incomplete Directed Perfect Phylogeny~\cite{Sha}.

\begin{definition}[Extended Matrix]
\label{def:Extended-Matrix}
Let $(M, E^*)$ be an instance of the CP-PP problem. The \emph{extended matrix}
associated with $M$ is an ${n \times 2m}$ matrix \me  over alphabet $\{0,1,?\}$
 which is obtained by replacing each column $c$ of $M$ by a pair of columns
$(c^{+}, c^{-})$.
Moreover for each   row $s$  of $M$
% \begin{enumerate}
% \item
if $M[s,c] = 1$, then  $M_e[s,c^{+}] = 1$ and $M_e[s,c^{-}] = 0$,
% \item
while if $M[s,c] = 0$ and $(s,c) \not \in E^*$, then  $M_e[s,c^{+}] =?$ and $M_e[s,c^{-}] = ?$,
otherwise $M_e[s,c^{+}] =0$ and $M_e[s,c^{-}] = 0.$
% \end{enumerate}
\end{definition}

% \begin{figure}[htbp]
% %\subfigure[A binary  matrix $M$ and the extended matrix of size $n \times 2m$]
% \begin{center}
% \includegraphics[width=0.90\textwidth]{M_Me}
% \end{center}
% \caption{A binary  matrix $M$ and the extended matrix of size $n \times 2m$}
% \label{fig:standard-tree2}
% \end{figure}

In this case the characters $(c^{+}, c^{-})$ are called conjugate.
Informally, the assignment of the \emph{conjugate} pair $(?,?)$ in a species
row $s$ for two conjugate characters ($c^{+}, c^{-}$) means that
character $c$ could be persistent in species $s$, \ie it is first gained and
then lost.
On the contrary, the pair $(1,0)$ means that character $c$ is only gained by
the species $s$. Finally, the pair $(0,0)$ means that character $c$ is never  gained  by
the species $s$.
A \emph{completion}  of  a pair $(?,?)$  associated to a species $s$ and
characters ($c^{+}, c^{-}$) of  \me  consists
of forcing $\me [c^{+},s]=\me [c^{-},s]=0$ or  $\me [c^{+},s]=\me [c^{-},s]=1$.

% (notice that the latter case corresponds to the $2$ state of the GCC  Problem
%for non-branching character trees with 3 states~\cite{DBLP:conf/isbra/ManuchPG11})

A (partial) \emph{completion}  \me is a completion of some of its conjugate
pairs.
% A completion is \emph{full} if all its conjugate pairs are completed.
A fundamental result  states  that $M$ admits a persistent phylogeny  if
and only if there exists a completion  of $M_e$ admitting a directed perfect
phylogeny~\cite{DBLP:journals/tcs/BonizzoniBDT12}.
Now, the same result holds when the input is a constrained matrix  $(M,E^*)$.
In fact,   if $(M,E^*)$ admits a solution, it means that it is possible to build a p-pp tree $T$  where
  each $(i,j)$   in $E^*$ implies that $c_j$ is not persistent  in species $s_i$ of tree  $T$.
As a consequence of this fact, there exists a completion of $M_e$ where
$\me [{c_j}^{+},{s_i}]=\me [{c_j}^{-},{s_i}]=0$. Moreover, by interpreting a conjugate pair of characters as characters of $\me$ is immediate to verify that
tree $T$ is a directed perfect phylogeny for  $\me$. Vice versa, given an extended matrix $\me$ that admits a perfect phylogeny $T$, then it is possible to show that $T$ is a persistent phylogeny for $M$.   This fact is an immediate consequence of definition \ref{def:Extended-Matrix} and the result in \cite{DBLP:journals/tcs/BonizzoniBDT12}.

%In the paper, we give polynomial time solutions for a special case  of the  CP-PP problem and a fixed parameter algorithm, thus our results partially answer the open problems.

\subsection{The conflict graph and red-black graph reduction}

% Let us recall that the directed perfect phylogeny is a restricted
% case of the persistent phylogeny where  no persistent characters are
% allowed~\cite{Gus91}.

% \begin{theorem}
% \label{equivalence}
%   Let $M$  be a binary matrix and let $M_e$ be the extended matrix   associated
% with $M$. Then $M$ admits a p-pp tree if
%   and only if there exists a completion  of $M_e$ admitting a directed perfect
% phylogeny.
% \end{theorem}

Let $M$ be a binary matrix.
Given two characters $c_{1}$ and $c_{2}$, the configurations induced by the
pair $( c_{1}, c_{2} )$ in matrix $M$ is the
set of ordered pairs $( M[s,c_{1}], M[s,c_{2}])$ over all species $S$.
Two characters $c_{1}$ and $c_{2}$ of $M$ are \emph{conflicting} if and only if
the configurations induced by such pair
of columns is the set of all possible pairs $( 0,1)$, $( 1,1) $, $(1,0) $ and
$( 0,0) $. Notice that being in conflict is a symmetric relation, therefore we
can define the \emph{conflict graph}  $G_c =( C,E_{c}\subseteq C \times C)$ of
a matrix $M$,  where a pair
$(c_{i}, c_{j}) \in E_{c}$ if and only if  $c_{1}$ and $c_{2}$ are conflicting
in $M$.
Notice that a conflict graph with no edges (called \emph{empty}) does not
necessarily imply the existence of a rooted perfect phylogeny, because of the
occurrence of the forbidden matrix with only the three configurations $(1,1)$,
$(1,0)$ and $(1,0)$.
However, by allowing a character to be persistent, the matrix admits a rooted
persistent perfect phylogeny.
We also need some graph-theoretic definitions.
A graph is called \emph{edgeless} if it has no edges.
A connected component is called \emph{nontrivial} if it has more than one
vertex.

\begin{figure}[htbp]

\begin{center}
\includegraphics[width=0.50\textwidth]{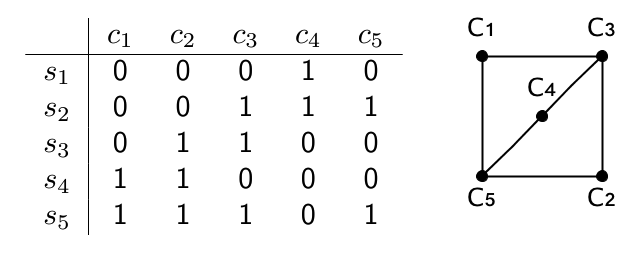}
\end{center}

\caption{A matrix and its conflict graph}
\label{fig:standard-tree}
\end{figure}

The first step is to simplify, if possible, the instance.
In fact we can always remove duplicate rows or columns.
Moreover a character $c$ is called \emph{universal} if all species have $c$,
while $c$ is called \emph{void} if no
species has $c$.
Again, we can always start by removing void characters.
In the rest of the paper, by an extended matrix \me we mean a matrix that may be partially or fully completed.
Besides the conflict graph, we introduce a second graph, called \emph{red-black
graph} and denoted by \grb, which will
be fundamental in our algorithm.

%\notaestesa{PB}{possiamo mantenere le definizioni e modificare dopo come segue}
% \noindent
% {\bf The red-black graph and the completion of matrix $M_e$}

\begin{definition}[Red-black graph]
\label{definition:red-black-graph}
Let \me be an extended matrix.
Then the \emph{red-black graph} $\grb=\langle V,E \rangle$  associated to \me
is the edge-colored graph where (i)
the vertices are the species and the conjugate pairs of \me (that is
for each two conjugate
characters $c^{+}$ and $c^{-}$, only $c$ is a vertex of \grb), (ii) a pair $(s,c)$
is a black edge iff the conjugate pairs
$c^{+}$ and $c^{-}$ are still incomplete in matrix \me and
$M_e[s,c^{+}]=1$ and
$M_e[s,c^{-}]=0$, (iii)
$(s,c)$ is a red edge iff $M_e[s,c^{+}]=M_e[s,c^{-}]=1$.
\end{definition}

\subsection{Realizing characters  in a red-black graph}

In the following we describe some completions of conjugate pairs of an extended
matrix \me that can be expressed as graph operations over the red-black graph
\grb
associated to the extended matrix.
Let $(c^{+}, c^{-})$ be two conjugate characters of \me, and let $\Cc$ be the
connected component of \grb containing the vertex $c$.
Given a partially completed matrix or, equivalently, a red-black graph, a
character is in one of three possible states: inactive (the initial state of all
characters), active, and free.
The  \emph{realization} of a  character $c$  in  \grb  consists of the
following steps:
\begin{enumerate}
  \item
    if $c$ is inactive  and for all species $s \in \Cc$, $M_e[s,c^{+}]=M_e[s,c^{-}]\not=0$ then:
  \begin{enumerate}
  \item
    for each species $s \notin \Cc$, pose  $M_e[s,c^{+}]=M_e[s,c^{-}]=0$;
  \item
    for each species $s \in \Cc$  if $(c,s)$ is not an edge of \grb, add
    a red edge $(c,s)$  and complete $\me$
    by posing  $M_e[s,c^{+}]=M_e[s,c^{-}]=1$;
  \item
    remove from \grb all black edges $(c,s)$ and label $c$ \emph{active}.
  \end{enumerate}
\item
else if $c$ is active and $c$ is connected by red edges to all species in
  ${\cal C}(c)$,  then:
  \begin{enumerate}
  \item
    all such red edges are deleted from \grb and $c$ is labeled \emph{free};
  % \item
  %   complete all remaining  pairs $(M_e[s,c^{+}],M_e[s,c^{-}])$, for $s
  %   \not\in \Cc$    by posing $M_e[s,c^{+}]=M_e[s,c^{-}]=0$.
  \end{enumerate}
\end{enumerate}

In some cases ($c$ is free, or $c$ is active but there exists a species $s\in \Cc$ that is not connected to $c$ by a red edge, or $c$ is inactive but persistent for a species $s\in
\Cc$, i.e. $M_e[s,c^{+}]=M_e[s,c^{-}]=0$) none of the stated conditions
hold, therefore the realization is \emph{impossible}.
Notice that realizing a character corresponds to a partial completion of the
matrix \me that is called the
\emph{canonical completion} of $c$ in \me, which in turn corresponds to the
construction of a standard tree, as follows.

% The following notion given in~\cite{DBLP:journals/tcs/BonizzoniBDT12} is
% crucial in providing the main result of this paper.

\begin{definition}[Standard tree]
\label{definition:standard-tree}
  Let $\me$ be an extended matrix, then a \emph{standard}
p-pp~\cite{DBLP:journals/tcs/BonizzoniBDT12} solving $\me$ has
  the following additional properties:
  \begin{enumerate}
\item
  some edge $e=(v,w)$ may not   labeled by any character.
 In this case
$l_{v}=l_{w}$;
\item
  no internal node has more than one child that is a leaf;
\item
  each leaf is incident on an unlabeled edge;
\item
  each species labels a leaf $x$ of $T$ and the parent of $x$ in $T$;
\item
  each internal node $x$ is labeled with the set $S(x)$ of species and
  with the set $C(x)$ of characters appearing in the subtree $T(x)$ of $T$
rooted at $x$;
\item
  some nodes $x$ are also labeled by a red-black graph $\grb(x)$ defined as
follows.
The graph associated to the root $r$ is the  red-black graph associated to the
extended matrix  \me and where all vertices are inactive.
 Each other node
$x$ might be labeled by the connected components of the red-black  graph having
species and characters of  tree $T(x)$ and
obtained  from graph $\grb(r)$ realizing in sequence all characters  (active or inactive) labeling
the edges on the path from $r$
to $x$ in $T$.
\item
Each internal node $x$ is labeled with the conflict graph $\gc(x)$ computed on
the submatrix of $M_{e}$ induced by the columns and rows in $C(x)$ and $S(x)$.
% Each other node $x$ of $T$ might be labeled by the connected components of the
% red-black  graph having species and characters of  tree $T(x)$ and
% obtained by starting from $\grb(r)$ and  realizing in sequence all characters labeling the edges on the path from $r$
% to $x$ in $T$.
  % \item
  % each node $x$ is also labeled by two sets $P(x)$ and $N(x)$ where $N(x)$ is the set of characters $c$ such that the
  % path of $T$ from $x$ to the root of $T$ has an edge labeled by $c^{-}$ and $P(x)$ is the set of characters $c$ that do
  % not belong to $N(x)$ and such that the
  % path of $T$ from $x$ to the root of $T$ has an edge labeled by $c^{+}$.
%in the same connected component as $c$ in the red-black graph
%that is obtained from the red-black graph associated to the instance of P-PP and realizing in sequence all characters
%labeling the edges of the path from the root of $T$ to  node $x$.
\item
Each subtree rooted at a node $y$ of $T$ whose parent $x$ is labeled by   the
red-black graph $\grb(x)$ has  species
 $S(y)$ and characters $C(y)$ such that $S(y)\cup C(y)$ is a union of connected
components of $\grb(x)$.
% %where the node connecting $x$ with its parent  $u$ is labeled by the
%character $c$,
% the set $S[x]$ of species that are descendants of $x$ in $T$ is equal to the
%set of species of the connected
% components of graph $\grb(x)$ having some species in $S[x]$.
\item
Let $x$ be a node such that the red-black graph $\grb(x)$ has connected
components
$C_{1}, \ldots , C_{k}$ with $k>1$ (\ie \grb is disconnected).
Then $x$ has $k$ children $x_{1}, \ldots, x_{k}$, where $S(x_{i})\cup C(x_{i}) =
C_{i}$ and each edge $(x, x_{i})$ are not labeled by any character.

\end{enumerate}
\end{definition}

It is easy to show that given a generic p-pp tree solving a matrix $\me$, we can modify it to satisfy    the  properties  from 1 to  7 of a standard tree.
In fact, these properties are obtained by adding a  labeling of nodes and leaves
of  the tree (see 4-5-6-7) or are obtained by assuming that there are not any identical rows (see 2)
or by adding additional edges (see 3).
Consequently, when proving that a p-pp tree is standard, we need to prove that properties 8-9 hold  under the assumption that the tree is in  the form
specified by properties 1-7.

\begin{figure}[htbp]
%\subfigure[Realization of character
%$c_4$]{
\includegraphics[width=0.55\textwidth]{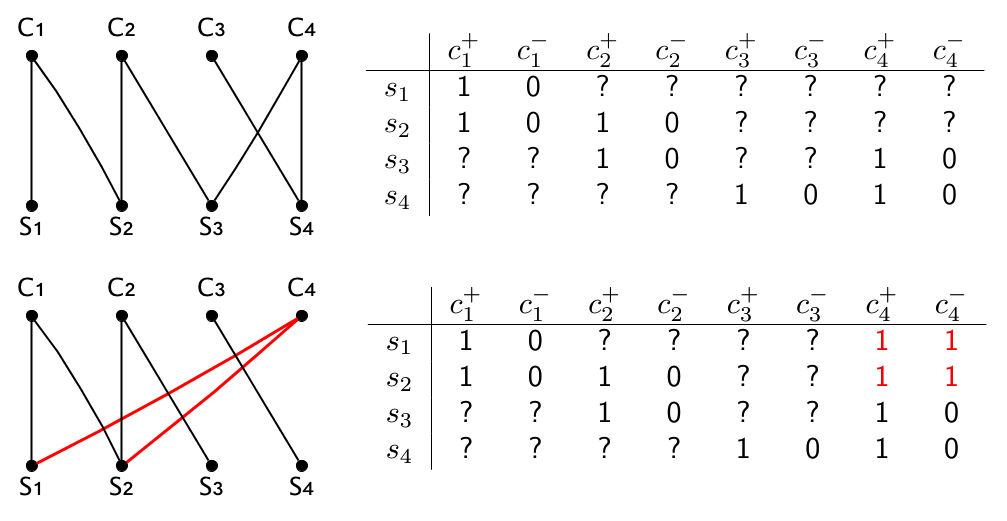}
%}
%\qquad
 %\subfigure[Realization of species
%$s_4$]{
\includegraphics[width=0.55\textwidth]{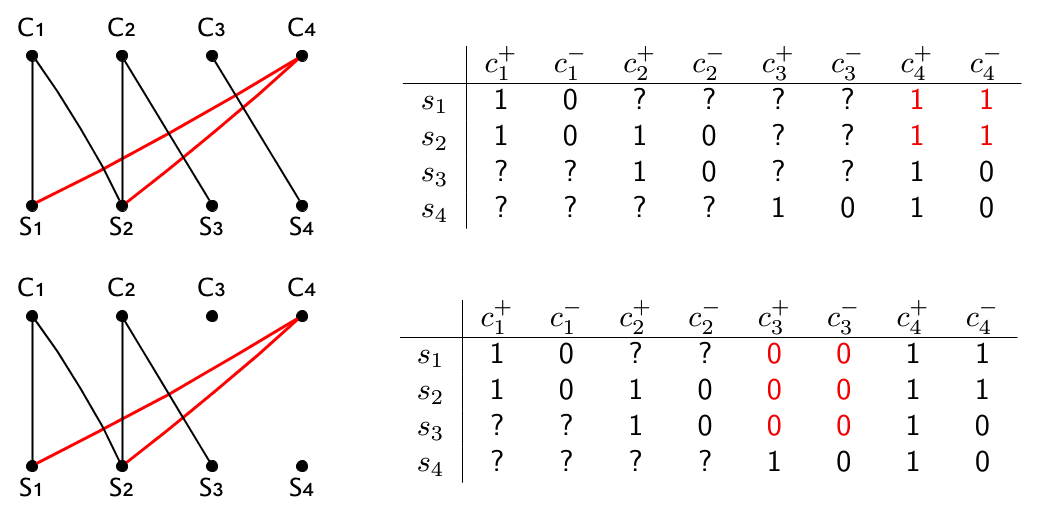}
%}

\caption{Figures (a)   and (b) illustrates the realization of a character  in a red-black graph   associated to an extended  matrix. The canonical
completion of the extended matrix  after the graph operations is shown for the
characters $c_4$ (a) and  $c_3$ (b).}
\label{fig:charct_species_realization}
\end{figure}

The following property of a standard tree  is central in characterizing matrices $\me$  that do not have  a solution.

\begin{property}
\label{property:sigma-standard-form}
  Let $T$ be a standard tree for an extended matrix \me and let $x$ be a node
of $T$.
  Then the red-black graph $\grb(x)$ that labels $x$ does not contain an
induced simple path starting from a species
  vertex and consisting of four red edges (such an induced graph is called red
$\Sigma$-graph).
\end{property}

%In  property~\ref{property:realization-is-canonical-completion}, we call $\me'$ the incomplete matrix induced by  graph  $\grb'$, since it represents all rows and columns of matrix \me that still have %to be completed.

% \noindent
% Then the red-black   for an extended matrix $M_e$ is updated by removal of
% species or characters whenever a species is solved or a character is free in
% the graph.
% Then the  submatrix $M'$ of $M$ induced by a red-black  $G_{RB}$  consists of
% the matrix having all characters and species of   graph $G_{RB}$.

Since a red-black graph \grb implicitly induces a submatrix $M_{1}$ of $M$ whose
rows and columns are the species and characters of \grb, we can define the
the conflict graph induced by $G_{RB}$ as the conflict graph on such submatrix $M_{1}$.
%
% \begin{definition}[conflict graph induced by a red-black graph]
% \label{def:conflict-induced-by-red-black}
% Let \grb be the red-black graph for a partially completed  extended matrix
% $M_e$   for $M$.
% Then  the conflict graph {\em induced by $G_{RB}$} is  the conflict graph of
% the submatrix of $M$ induced by the characters and species of $G_{RB}$.
% \end{definition}
%
% % Observe that the red-black $G(x)$ represents the submatrix that still has to be completed into a matrix realizing subtree $T(x)$.
% Such notion is important also in our paper, since we will describe an algorithm
% that computes a standard p-pp, if it exists.
% To transform a standard p-pp into an actual persistent perfect phylogeny means
% to contract all edges that have no labels
% and to remove all superfluous labels.
%
% For simplicity in the paper we will assume that each edge of the standard tree
% is labeled by at most one character.
%
Observe that  for every   descendant  $x_i$ of node $x$, then by
Definition~\ref{definition:standard-tree} of standard tree, there will be a
connected component of $\grb(x)$ whose species are
exactly those of $T(x_i)$.
Moreover, given an internal node $x$ of a tree $T$, the \emph{label sequence}
of $x$ is the sequence of characters of the edges on
the unique path from the root of $T$ to $x$.
Then we associate to $x$ the matrix $\me(x)$ that is the  partially completion
of \me according to the label sequence of $x$ where only the rows and columns
corresponding to vertices of $\grb(x)$ are retained.

\begin{property}
\label{prop:outgoing-edges-from-char-all-same-color}
  Let $G_{RB}(x)$ be the red-black graph  associated  to a node $x$, let $c$ be
any character of  $G_{RB}(x)$, let $e^{+}$ be the edge
  of $T$ labeled by $c^{+}$ and let $e^{-}$ be the edge   of $T$ labeled by
$c^{-}$ (if it exists).
  By construction of the red-black graph, only three cases are possible:
  \begin{enumerate}
  \item
    $c$ is inactive, all edges of $G_{RB}(x)$ that are incident on $c$ are
black, and $c^{+}$
    does not belong to the label sequence of $x$;
  \item
$c$ is active,    all edges of $G_{RB}(x)$ that are incident on $c$ are red,
moreover $c^{+}$
    belongs to the label sequence of $x$, but $c^{-}$
    does not;
  \item
$c$ is free,    $c$ is an isolated vertex of $G_{RB}(x)$.
  \end{enumerate}
\end{property}

Using the stated properties  we can prove the characterization
stated in    Theorem~\ref{main-equivalence} used to test   the existence of a  solution of the CD-PP problem.
We discover that the Theorem  builds upon an analogous characterization  for  the P-PP problem~\cite{DBLP:journals/tcs/BonizzoniBDT12}, stating that an extended matrix $M_e$ has a perfect phylogeny (\ie a pp tree)    if and only if it has a   standard tree. Moreover a standard tree represents a canonical completion of all characters of the matrix, i.e. the realization of a sequence of all characters,  called
\emph{c-reduction} that leads to an edgeless red-black graph (see Definition \ref{c-reduction} given below).
%An important aspect is that a solution $T$ of the CP-PP problem can be constructed
%by considering a sequence of
%characters, called a \emph{c-reduction}.
%A c-reduction  is \emph{successful} if the realization of its corresponding characters one after the other
%results in an edgeless red-black graph.

\begin{lemma}
\label{lemma:2-tcs-part1}
Let $M_{e}$ be an extended matrix admitting a perfect phylogeny. %persistent
Then there exists a perfect phylogeny $T$ realizing $M_{e}$ such that for each
node  $x$ of $T$ and for each
character $c_{1}$, if  $c_{1}$ is adjacent to two species $s_{1}$, $s_{2}$ in
$G_{RB}(x)$, then  $s_{1}$, $s_{2}$ are either both in $T(x)$ or none of them is
in $T(x)$, where $G_{RB}(x)$ is the red-black graph obtained by realizing the
label sequence of $x$ from the red-black graph associated to \me.
\end{lemma}

Now Lemma \ref{lemma:2-tcs-part1} is used to prove the following lemma, first proved in~\cite{DBLP:journals/tcs/BonizzoniBDT12},
showing  that when an extended matrix admits a solution, than it must have a  standard tree.
Moreover its proof establishes
a correspondence between a  standard tree and the red-black graphs associated
to nodes.
More precisely,  given a standard tree $T$, such a tree  represents the
canonical completion of $M_e$ that we obtain by visiting tree $T$ and realizing
 characters of  visited edges of tree $T$ in the red-black graph  for $M_e$.
Then, when reaching  a node $x$ of $T$ by visiting tree $T$, the red-back graph
$G_{RB}(x)$ represents the species and characters that still have to be
completed in the extended matrix.

\begin{lemma}
\label{lemma:2-tcs}
Let $M_{e}$ be an extended matrix admitting a p-pp tree. Then, there exists a
 completion  of $M_{e}$   that is realized by a
standard tree $T$.
\end{lemma}

Consider an extended matrix \me that has been completed  and \me admits a standard tree $T$.
Let us consider the sequence of character corresponding to a depth-first visit of
$T$ and the corresponding red-black graph $\grb^{*}$ obtained from \grb by
realizing all characters in the visit.
First of all, notice that $\grb^{*}$ cannot contain any black edges, for
otherwise we would immediately contradict the hypothesis that $T$ solves \me.
Moreover we can prove that $\grb^{*}$ cannot contain any red edge $(c,s)$.
Assume to the contrary that such an edge exists.
Since $T$ solves \me, the path of $T$ from the root to the leaf labeled by $s$
contains the edge labeled by $c^{+}$ and the one labeled by $c^{-}$.
Let $(x,y)$ be the edge labeled by $c^{-}$.
By point 9  of Definition~\ref{definition:standard-tree}, the red-black graph
$\grb(y)$ is connected, therefore all red edges connecting $c$ and a leaf of
$T(y)$ are deleted.
Since $s$ is a leaf of $T(y)$, the red edge $(c,s)$ is deleted.

% definition it must be that no black edge is given in the  red-black graph.
% Moreover, it is easy to show that no red-edge is left and hence the red-black
% graph is edgeless. This fact is a consequence  of the well known
% characterization of matrices admitting a perfect phylogeny as the one  having no
% forbidden matrix,  where the forbidden matrix induced by negated characters is
% represented by the red $\Sigma$-graph.\notaestesa{GDV}{ Non mi convince che sia
% una conseguenza diretta della caratterizzazione classica.
% Secondo me è necessaria anche un'analisi della costruzione del grafo rosso-nero.
% la proprietà vale.
% }
%  In fact, the only way to have red-edges in a reduced graph is having a path of length four in the graph, \ie the red  $\Sigma$-graph.  Notice that if a red-black graph has only red-edges and  a red  $\Sigma$-graph is forbidden then the graph must be  edgeless.

% Property~\ref{property:sigma-standard-form}  easily follows from the
% characterization of solvable matrices.
% %----da contrallare---.

\begin{comment}

\begin{property}
\label{edgeless}
Let $T$ be a  standard  tree that solves an extended matrix $M_e$,  let $x$ be a
leaf of $T$ and let $G_{RB}(x)$ be the red-black graph of $x$.
Then graph  $G_{RB}(x)$ is edgeless.
\end{property}

\end{comment}

\begin{definition}
\label{c-reduction}
Let \grb be a red-black graph for an extended matrix $M_e$ and let $R$ be a
sequence of characters such that  each negative character $c^{-}$ appears after
the corresponding positive character $c^{+}$, and all characters of \grb appear
in $R$.
Then $R$ is a  \emph{successful c-reduction of} \grb
if and only if all single character realization  of $R$ are possible and the graph
reduced by $R$ has no edge.
Then we   say that \me is solvable.
\end{definition}

\begin{lemma}
\label{main-reduction0}
Let  $T$ be a  standard tree that solves a  matrix $M$ and let $V$
be the sequence of characters labeling the edges of $T$ according to a
depth-first visit.
Then $V$ is a successful c-reduction of $\me$.
\end{lemma}

By Lemma~\ref{main-reduction0}   a  standard tree is represented by a canonical completion of all characters, i.e. a c-reduction that leads to an
edgeless red-black graph.
Now  we can  show that having such a edgeless red-black graph   is not only a
necessary but a  sufficient condition for having a canonical completion
admitting a tree.

We have shown that  an extended matrix admits a persistent phylogeny
if and only if it has a directed perfect phylogeny.
Since the  matrices admitting a perfect phylogeny are those that do not contain
a forbidden submatrix~\cite{Gusfield}, we know that  an
extended matrix admits a persistent phylogeny if and only if  there exists a
completion such that no forbidden matrix induced by two negated characters, two
positive characters,  or one of each kind, is induced by the completion.
Let us show that having an edgeless red-black graph  due to a successful
reduction implies that the associated canonical completion $M'$  of $M_e$ has
no forbidden matrix.
First, observe that no forbidden matrix between two positive characters in $M'$
 is possible, since completing or realizing  two  characters ${c_1}^+$ and
${c_2}^+$ that are in the same connected component of the red-black  implies a
containment relation between the two characters (note that two characters in
disjoint components of the red-black graph cannot be connected during the
successful c-reduction).
Finally,  observe that the only way to have red-edges in a reduced graph is
with a path of length four in the graph, \ie the red  $\Sigma$-graph, corresponding to the forbidden matrix induced by negative characters.  Notice
that if a red-black graph has only red edges and  a red  $\Sigma$-graph is
forbidden, then the graph must be  edgeless.

%More precisely, we can  state the following fact that uses
%Lemma~\ref{lemma:2-tcs}.

\begin{theorem}
\label{equivalence-0}
  Let $M$  be a binary matrix and let $M_e$ be the extended matrix   associated
with $M$ such that its associated red-black graph admits a successful c-reduction.
Then $M_e$ is solvable by a standard tree.
\end{theorem}

Lemmas~\ref{equivalence-0} and~\ref{main-reduction0} allow to generalize the   Theorem initially proved  in~\cite{DBLP:journals/tcs/BonizzoniBDT12} for the P-PP problem.
\begin{theorem}
\label{main-equivalence}
Let $M_e$ be an instance of the  CP-PP problem. Then $M_e$ is solvable if and
only if there exists a successful c-reduction for the red-black graph $G_{RB}$
associated to $M_e$.
\end{theorem}

Therefore we only need to find a successful c-reduction; if none exists, then
the initial instance has no solution. Finally, observe that the property 8 and 9 of a standard tree imply that  we need to find a successful c-reduction for each connected component of the red-black graph, and any concatenation of those  successful c-reductions on a single component gives a successful c-reduction  for the whole graph.

\section{Solving the CP-PP problem in polynomial time  on matrices with empty conflict graphs}

In the following, by using properties of the red-black graph,  we show that a persistent perfect phylogeny always exists for a matrix $M$ that has  an empty  conflict graph.
The main characterization given by Theorem \ref{main-equivalence} is used to solve the CP-PP problem, since we design a procedure that finds a successful $c$-reduction of the red-black graph associated to  the extended matrix for $M$.

Given $M$ a binary matrix,   the {\em partial order graph}  for $M$  is the partial order  $P$   obtained by ordering   columns of $M$ under the  $<$ relation which is defined as follows:
character $c < c' $ if and only if  $M[s, c] \leq M[s, c']$ for each species $s$, otherwise we say that $c$ and $c'$ are not comparable. Moreover, we build a  graph $G=(V,E)$, called {\em adjacency graph} for $M$:   $V$ is the set of columns of $M$ and $(u,v)$ is an edge of $G$ if and only if
$u, v$ are \emph{adjacent}, i.e. there is a species $s$ that is adjacent to both $u$ and $v$
in the red-black graph for the extended matrix $M_e$  associated with $M$. Then the following result holds.

  \begin{lemma}
  \label{clique}
 Let $M$ be a binary matrix with  an empty conflict graph.  Assume that the extended matrix associated with $M$ induces a connected red-black graph and let $P$ be the partial order graph for $M$.
Let $C_M$  be the set of maximal elements in $P$.
Then $C_M$ consists of  elements that are pairwise adjacent in the adjacency graph.

%Then set $C'$   forms  a clique in the adjacency graph  $G$ for $M$.
  \end{lemma}

  \begin{lemma}
  \label{consecutive-p}
Let $M$ be a binary matrix that has an empty conflict graph. Let $G_{RB}$ be the red-black graph for the extended matrix associated with $M$.  The realization of two characters $a$ and $b$  that are adjacent in the adjacency graph for $M$ produces at most two red disjoint components, one  with only vertex $a$ and the other with only vertex $b$.
  \end{lemma}

The above  Lemma \ref{consecutive-p} combined with Lemma \ref{clique}  implies that maximal characters can be realized in an arbitrary order and are crucial to prove the following Theorem.

\begin{theorem}
\label{pol_empty}
Let $M$ be a binary matrix that has an empty conflict graph.
Then $M$ admits a persistent perfect phylogeny $T$ and there exists a polynomial
time algorithm to compute $T$.
\end{theorem}

If we consider a constrained matrix $(M, E^*)$ where $E^* \not= \emptyset$ and having an empty conflict graph, the
CP-PP problem might not have a  solution for $(M, E^*)$ since some characters cannot be realized.
In fact, if all characters are not persistent, we obtain the classical perfect
phylogeny problem, and even this problem might not have a solution  when the conflict graph is empty, since the matrix $(M, E^*)$ could contain the forbidden matrix.
However, Lemmas~\ref{clique},~\ref{consecutive-p}  and Theorem~\ref{main-equivalence} allow  to prove the correctness of the following procedure that can be used to solve the CP-PP in polynomial time in the case of empty conflict graph. In fact,   the two Lemmas hold also for constrained matrices $(M, E^*)$ (see proof in the Appendix).

%In the following we give a procedure for the reduction of red-black graph of a matrix having an empty conflict graph.

\noindent
{\bf Procedure Solve-CP-PP-empty-conflict$(M, E^*)$}

 {\em Input:} a constrained  binary matrix $(M, E^*)$ that has  an empty conflict graph.

 {\em Output:} a  realization $S_c$ of characters  to successfully reduce graph \grbd

\begin{itemize}

\item Build the partial order  $P$ for   $M$; let $G_{RB}$ be the red-black graph for the extended matrix for $M$.

\item let $C_M$ be the set of all maximal elements in $P$ that are in the same connected component of the red-black graph $G_{RB}$; then realize in an arbitrary order  those characters in $C_M$ that can be realized.
Repeat the step till all characters  have been realized. Otherwise, return {\em no solution}.

\end{itemize}

Now, the correctness of the  algorithm  for the CP-PP problem is again a consequence of the fact that maximal characters can be realized in an arbitrary order and the fact that maximal characters are realized before characters they include by the $<$-relation.
%In fact,  we   use also the fact that in a successful c-reduction of graph \grb  a maximal character $c$ can be realized before a character $c'$ such that  $c' < c$.
This fact relies on the observation that if $c < c'$ and $c'$ is not persistent in a species $s$, then also $c$ is not persistent in the same species.
In fact, if $c'$ occurs in the p-pp tree after $c$ then,  assuming that $c$ is persistent in a species $s$, also $c$ must be persistent in the same species, since $c'$ occurs before the negated character $c^-$.
Otherwise, if $c'$ occurs in the tree before $c$, then all species below $c^+$ and the negated character $c^-$  are also in $c'$, thus forbidding to have $c$ persistent in the same species of $c$.   The correctness of the  algorithm is proved in the Appendix.

 \section{An optimized algorithm}

 In this section we propose an algorithm for the CP-PP  problem called {\bf  \pp-opt}  that is based on the  procedure  {\bf Solve-CP-PP-empty-conflict$(M)$}.
Just as the  algorithm in~\cite{DBLP:journals/tcs/BonizzoniBDT12},  {\bf
  \pp-opt}  reduces an instance $M$ of P-PPH to an instance $\me$ of the  IP-PPH
problem. In fact,we know that $\me$  admits a pph tree $T$ if and only if $T$ is
a solution of  matrix $M$.
Then, by  Theorem~\ref{main-equivalence},  $\me$ admits a pph tree $T$ if and
only if there exists a successful c-reduction of the red-black graph for $\me$.
The algorithm in~\cite{DBLP:journals/tcs/BonizzoniBDT12}    explores all
permutations of  the set $C$ of  characters  of $\me$ in order to find one that
is a successful c-reduction, if it exists.
On the other hand,  {\bf  \pp-opt}  builds  a decision tree, where each edge   represents a character and each path of the tree from the root to a leaf is a   permutation $\Pi(C')$ of  a subset $C'$ of the set  $C$ of characters such that  the realization of $\Pi(C')$ makes the conflict-graph induced  by the red-black graph empty. Thus  {\bf  \pp-opt} strongly relies on the polynomial time solution of the P-PP problem in the case  the conflict-graph is  empty.
The algorithm works in a branch and bound like manner, in the sense that if a branch of the decision tree ending in node $x$ does not lead to a solution, then the decision tree below $x$ is discarded.  More precisely,  each branch ending in node $x$   gives a partial permutation  $\pi$  that consists of all characters labeling the path from root $r$ to  node $x$. A partial completion  $M_{\pi}$ is computed by realizing characters provided by the partial permutation $\pi$. Whenever $M_{\pi}$ contains the forbidden matrix,  then the branch ending in $x$ does  not lead to a solution, and $x$ is labeled as a {\em fail} node.
The algorithm either finds the first permutation of characters that provides a successful c-reduction of the red-black graph for the extended  matrix, if it exists, or will decide that the matrix does not admit a pph tree by visiting the whole decision tree.

Let us  describe the recursive procedure {\bf \pp-opt} which is initial invoked
as  {\bf \pp-opt ($M,  Me, r, \{r\}$)}, where $r$ is the root of the decision
tree  and the visited tree is the set $\{r\}$.

{\bf Algorithm \pp-opt($M$, $M'$, $x$, ${\mathcal T}$)}\\
{\em Input:} a binary constrained matrix $(M, E^*)$ of size $n \times m$, the
set $E^*$ of constraints, a partial depth-first visit tree ${\mathcal T' }$ of the decision tree ${\mathcal T}$  and a leaf  node $x$ of ${\mathcal T}$,  a partial completion $M'$ of the extend matrix \me obtained by the realization of  the  characters labeling a path $\pi$ from $r$ to node $x$ of the tree ${\mathcal T' }$;\\
{\em Output:}  the tree ${\mathcal T' }$ extended with the depth-first visit of ${\mathcal T }$ from node $x$. The procedure eventually outputs a successful c-reduction $r$ or   fails to find such a successful c-reduction.
\begin{itemize}
\item[-] Step 1: if the edge incident to node $x$ is labeled $c$ and  $c$ is
admissible, then realize $c$ in $G_{RB}$ and complete the pair of columns $(c,
c')$ in $M'$.
If the matrix $M'$ has a forbidden matrix, then label $x$ as a fail node. If $x$ is a leaf node, then mark $x$ as a successful node and output the permutation labeling the path from
the root $r$ of tree $T$ and the leaf $x$ of ${\mathcal T' }$.
\item[-] Step 2: compute the conflict graph $G_c$ for the matrix $M$, updated after the realization of the characters along the path $\pi$ from the root $r$ to node $x$ (i.e. $M$ is obtained after eliminating the rows that correspond to species-nodes  that are singletons in $G_{RB}$),
\item[-] Step 3: if the conflict graph $G_c$  is empty, then apply the polynomial-time algorithm for an empty conflict graph  and return a successful c-reduction, if it exists. Else for each node $x_i$ that is a child of node $x$ in tree ${\mathcal T' }$ and is labeled by a non-active character in $G_{RB}$, apply \pp -opt($M$, $M'$, $x_i$, ${\mathcal T' } \cup \{x_i\}$).
\end{itemize}

Notice that testing whether a character may be realized in a connected component of the red-black graph $G_{RB}$ requires time that is linear  in the number of species of the component.
Clearly, the algorithm requires to compute the connected components of $G_{RB}$, which can be done in time $O(f(n, m))$, where $f(n, m)$ is polynomial in the size of graph $G_{RB}$.
Consequently, the   time  required to evaluate a single path of the tree is $O(f(n, m) m)$, since the path may have $m$ characters to be realized and completed.
The total number of explored paths is clearly equal to the number of permutations of set $C$ of characters on the input matrix, in the worst case.

\section{Experimental analysis}

We have implemented the algorithm  {\bf  \pp -opt}  and tested it over simulated data produced by  the tool  \textit{ms} by Hudson \cite{Hu} and  on real data coming from the International HapMap project, a multi-country research to identify and catalog genetic similarities and differences in human beings~\cite{frazer_second_2007}.
The main goal of these preliminary experiments has been to test the average running time   when the number of rows and columns increase, since we wanted to test  the applicability of the method on matrices of given sizes.
The experimental analysis on real haplotype data aims to investigate the use of the persistent model to detect  haplotypes data that cannot be explained by the perfect phylogeny model but can exhibit persistent characters.

We have implemented the algorithm in C++ and the experiments have been run on a standard Windows workstation with 4 GB of main memory.

The results of a first experiment are reported in
Table~\ref{tab3}. The table   reports  the computation time to solve sets of $50$ matrices
for each dimension $50 \times 15$---\ie $50$ species and $15$ characters---
$100  \times 15$,  $200  \times 15$, and  $500  \times 15$ with a recombination
rate $1/15$.
The sets contain only matrices that are solved within five  minutes.
Clearly, the number of unsolved matrices increases with the size of the input matrices but also with the number of conflicts  that are present in the conflict graph.
The table also reports the results obtained by  comparing the execution times of
the exact algorithm given in~\cite{DBLP:journals/tcs/BonizzoniBDT12} ({\bf  \pp})
with the optimized algorithm on sets of matrices with a fixed number of columns
and different numbers of rows.
Since \pp works on unconstrained matrices, we have not introduced any
constraints in the input matrices.
The {\bf  \pp -opt} algorithm is able to find a solution for all matrices in contrast to the {\bf  \pp} algorithm that in some cases takes more than 10 minutes to find a solution for a single matrix.

  The average execution time to solve  $10$ matrices with a single conflict is of  0.031s, 	0.047s, 	0.093s for matrices of size 100x15,  200x15, 500x15 respectively.

%Another experiment has been done with 10 matrices of the same size $50 \times 15$ and different number of edges in the conflict graph. The average time was  $0.015$, $0.031$ and  $0.051$, %respectively for the case of $1, 5$ and $10$ conflicts.

%In order to test the performance of the algorithm for large matrices in terms %of number of species we have processed a matrix of size $1000 \times 15$ with a %conflict graph having $9 $ conflicts (edges).  It took $35.5$ seconds to find %the solution to the matrix. Moreover, we have processed a matrix of size $1000 %\times 30$ with a conflict graph having $149$ conflicts (edges).  It took %$604.66$ seconds to find the solution to the matrix.

In order to test the performance of the algorithm for large matrices in terms of
number of species we have processed a set of $20$ matrices, $10$ of size $1000
\times 30$ and $10$ of size $1000 \times 40$, with conflict graphs having a
number of conflicts (edges)  between $40$ and $236$. We fixed a maximum time
(900 seconds), after which the execution is stopped. In all cases the simulation
terminates without finding a solution, i.e.   all matrices do not admit a p-pp
tree.
Therefore those instances are likely to be the worst cases for our algorithm,
since the entire search tree is explored.
Considering the set of matrices of size $1000 \times 30$, in 40\% of tested
cases the simulation is stopped since the maximum time is elapsed. In the
remaining 60\% of the cases the simulation halts before the timeout.
In particular, in 10\% of cases we obtain a result in a time between 10 and 15
minutes, and in 50\% of cases the result is given in less than 3 minutes.
Considering the set of matrices of size $1000 \times 40$, in 50\% of tested
cases the execution is stopped because the timeout is reached.
In the remaining 50\% the result is given in a time between 3 and 6 minutes, with an average time of 216 seconds.

\begin{table*}
%\caption{Comparing the execution time of the optimized algorithm and the exact algorithm on matrices with fixed number of columns and different %numbers of rows. Each set contains 50 matrices of the same size.
\centering
\begin{scriptsize}
\begin{tabular}{|c|c|c|c|c|c|c|c|c|c|}
\hline
       nxm &   number  & total       & average   & \multicolumn{ 2}{c|}{solved matrices} & \multicolumn{ 2}{c|}{total time in s} & \multicolumn{ 2}{c|}{average time in s} \\
           &   P-PPH         &  conflicts  & conflicts & \pp~\cite{DBLP:journals/tcs/BonizzoniBDT12} & \pp -opt & \pp~\cite{DBLP:journals/tcs/BonizzoniBDT12} & \pp -opt & \pp~\cite{DBLP:journals/tcs/BonizzoniBDT12} & \pp -opt \\
\hline
     50x15 &          6 &        236 &       4.72 &         47 &         50 &      89.12 &      32.32 &       1.90 &       0.65 \\
    100x15 &          4 &        175 &        3.5 &         48 &         50 &     436.02 &     194.63 &       9.08 &       3.89 \\
    200x15 &          3 &        147 &       2.94 &         48 &         50 &    1583.50 &      43.21 &      32.99 &       0.86 \\
    500x15 &          7 &        219 &       4.38 &         44 &         50 &     888.59 &     889.43 &      20.20 &      17.79 \\
\hline
\end{tabular}
\end{scriptsize}
\caption{The table has  entries to specify  the average time to solve a single matrix (in seconds shortened as s),  the  number of matrices that do not admit a p-pp tree, the total number of conflicts, measured as the number of edges in the conflict graph  of the matrices of each set, and the average number of conflicts.
Each considered matrix has a conflict graph    that consists of a single non trivial component.}
\label{tab3}
\end{table*}

\vspace{-.2in}

\begin{comment}

\begin{table*}
\caption{Average execution time in seconds to solve  $10$ matrices with a single conflict. }
\label{tab2}
\centering
%\begin{scriptsize}
\begin{tabular}{|c|c|c|c|}
\hline
 {\bf 50x15} & {\bf 100x15} & {\bf 200x15} & {\bf 500x15} \\
\hline
0.015&	0.031&	0.047&	0.093 \\
\hline
\end{tabular}
%\end{scriptsize}

\end{table*}

To show how the number of columns influences the execution time we simulated a set of matrices with not empty conflict graph and no solution (in this way we are sure that the entire searching-tree is visited). Each matrix has 10 rows. The number of columns changes from 10 to 30. Fig.\ref{tempi} shows how the execution time increases in relation to the increasing number of columns.

\begin{figure}[t!]
\begin{center}
\includegraphics[width=0.8\textwidth]{tempi.pdf}
\end{center}
\caption{Execution time: exponential growth in relation to the number of columns}
\label{tempi}
\end{figure}
\end{comment}

%%%%%%%%%%%%%%%%%%%%%%%%%

Finally,  the algorithm has been tested on real data coming from the International HapMap project. The data are classified by type of population. In our case, we used data from the set ASW (African ancestry in Southwest USA). Each individual is described by the two haplotypes (in our application the two haplotypes correspond to two different species, i.e. two different rows of the matrix).
%The matrices in HapMap are described using the nucleotides alphabet {A, C, G, T}. To adapt this format to our algorithm, a pre-processing is needed to translate each matrix into a binary matrix. Since each %character (column of our matrix) is described by a maximum of two different nucleotides, we considered the dominant character (i.e. the most frequent character in the considered column) as 0, and the %recessive character as 1.
The   data set consists of binary matrices of dimensions $10 \times 10$, $26
\times 15$, $26 \times 25$, and $26 \times 30$. For each group we considered 10
matrices. In all cases the matrices do not admit perfect phylogeny, and the
number of conflicts changes from a minimum of 4 to a maximum of 138.
Increasing the size of the matrix, and therefore the number of conflicts, the percentage of matrices that admit persistent perfect phylogeny  decreases. More in detail, 80\% of the tested matrices of size $10 \times 10$ admits solution, only 20\% of the tested matrices of size $26 \times 15$ admits solution, and none of the sets $26 \times 25$, and $26 \times 30$ admits solution.
The results confirm the conjecture that haplotype data may be related by the persistent phylogeny in case they cannot be explained by the perfect model. Moreover, the results on simulated and real data show the good performance in terms of time for matrices of a certain size even in the case that no solution is given, i.e. the searching space increases. This  behavior  is due to the fact that the algorithm includes a test to bound a path whenever the path does not lead to a solution, since we test whether a partial  completion includes the forbidden matrix.
%%%%%%%%%%%%%%%%%%%%%%%%%%

A final experiment an interesting case where we have matrices of size 30 x 60
with 4 or 5 conflicts (hence those matrices do not admit a perfect phylogeny).
On those matrices \textbf{\pp} have always reached the 15-minute timeout without
computing a solution.
On the other hand, imposing that a few (less than $10$) characters cannot be
persistent allows {\bf \pp -opt} to find a solution in a few seconds.
This experiment shows that introducing some constraints can help in finding a
solution, hence a feasible strategy to determine if a matrix has a persistent
phylogeny is to introduce some random constraints.
We are currently planning an extension of {\bf \pp -opt} that introduce some
deterministic constraints, based on an analysis of the initial conflict graph,
to speed up the computation.

 \section{Conclusions and open problems}

 %\vspace{-.2in}
 In this paper we have investigated  the CP-PP problem, which is the general problem of computing a persistent  perfect phylogeny for   binary matrices where some characters  may be forced not to be persistent.
 The  case where all characters may be persistent is called P-PP problem. The computational complexity of this
 problem is still open, except when the output is a specific tree, that is,  a branch or a path.
In order to find algorithmic solutions of the  CP-PP problem we  proved  that,
similarly to the P-PP problem~\cite{DBLP:journals/tcs/BonizzoniBDT12},  CP-PP
can be reduced to a problem, called {\em Red-black graph reduction problem}, of
emptying a graph with a sequence of operations on characters.
This formulation of the problem allows us to find  a  polynomial time solution
for matrices whose conflict graph is empty.
Based on this result we propose a fixed parameter algorithm (the parameter is the number of characters) that uses a branch-and-bound technique to reduce the computation time.

An experimental analysis over simulated matrices and on binary matrices
representing real haplotypes shows that the algorithm is able to process the
majority of matrices with 40 characters,  1000 species, and no solution in a few
minutes (clearly, for all those cases the algorithm determines that no solution
exists---those instances are usually the hardest for our approach), while being
very fast on matrices with  1000 species and  15--20 characters that admit a solution.

  Finally, we observe that the CP-PP problem  (P-PP problem) is equivalent to cases of the
 General Character Compatibility problem (GCC) investigated in ~\cite{DBLP:conf/isbra/ManuchPG11}, when  states are over set
 $\{0,1,2\}$ and when, for each generalized character $\hat{\alpha}$, ${\alpha}(s)
 \in \{\{1\}, \{0\}, \{0,2\}\}$  (respectively,  ${\alpha}(s) \in \{ \{1\}, \{0,2\}\}$ for the P-PP case ) for each species $s$ and $T_{\alpha}
 =  0 \rightarrow 1 \rightarrow 2$.
It is interesting to note that the computational complexity of the  cases of the GCC problem that are equivalent to  P-PP and CP-PP,
%---where for each generalized character $\hat{\alpha}$ it holds that $ {\alpha}(s) \in \{\{1\}, \{0\}, \{0,2\}\}$  for each species $s$ and $T_{\alpha}
% =  0 \rightarrow 1 \rightarrow 2$
i.e.  cases 5 and 6 of Table~1 in ~\cite{DBLP:conf/isbra/ManuchPG11},  is still
open,  while only partial results are obtained   when the solution is constrained to be a branch or a path (see the Appendix).
Thus the results we give in the paper also apply to those cases.

\subsection{Acknowledgements}
The authors would like to thank Rob Gysel for pointing out the connection of the persistent perfect phylogeny with the GCC problem.
Anna Paola Carrieri also thanks Dan Gusfield for useful discussions on the problem.

 %It is interesting to note that the computational complexity of the above  two cases of the GCC problem that are special cases of   the P-PP and CP-PP problem are still open
%(they correspond to cases 5 and 6 of the Table  1 in ~\cite{DBLP:conf/isbra/ManuchPG11}), while only partial results are obtained in the case the solution is constrained to be a branch or a path.

% \bibliographystyle{abbrv}
% \bibliography{biblio-ppp_algorithm,abbreviations,books,graphs,biology,complexity,haplotype}

\newpage
\appendix

\section*{Appendix}

%Clearly these results  hold for the CP-PP and P-PP problem.
%

{\bf Proof of Lemma \ref{lemma:2-tcs-part1}}

\begin{proof}
   Assume to the contrary that no such tree exists.
  More precisely, we will consider, among all realizations of $M_{e}$, a tree
$T$ with an  edge $e=(u,x)$ that is
  closest to the root among   all edges contradicting the lemma, that is
$s_{1}\in T(x)$,
  $s_{2}\notin T(x)$.
  Clearly $e$ is labeled by $c^{+}$ or $c^{-}$. Let $y$ be the least  common
ancestor of $s_{1}$ and $s_{2}$. Since  $s_{2}\notin T(x)$,
  $y$ must be an ancestor of  $x$.

  By  cases (1)  and (2) of
Prop.~\ref{prop:outgoing-edges-from-char-all-same-color}, since $c_1$ has
incident edges from $s_1$ and $s_2$ in $G_{RB}(x)$,  both edges
$(c_{1},s_{1})$, $(c_{1},s_{2})$ are black or red.
  Assume  they are both black in graph  $G_{RB}(x)$: thus they must be black
also in graph  $G_{RB}(y)$.
   For this reason, note that $c_1$ is adjacent to $s_1$ and $s_2$ also in
graph  $G_{RB}(y)$ as they are adjacent after the the realization of
  the characters in the label sequence of $x$.

Since both $s_{1}$ and $s_{2}$ have the character $c_{1}$, the edge $e_{1}$ of
$T$ labeled by $c_{1}$ must lie on the common portion of the paths from the root
to both $s_{1}$ and $s_{2}$ that is $e_{1}$
lies on the path from the root to node $y$ (as $y$ is the the least  common
ancestor of $s_{1}$ and $s_{2}$).
Therefore  both edges $(c_{1},s_{1})$, $(c_{1},s_{2})$ are red in graph
$G_{RB}(y)$ (by definition of graph associated to a node), contradicting the
previous assumption. Thus assume that $(c_{1},s_{1})$, $(c_{1},s_{2})$ are red
in $G_{RB}(x)$.

By case~2 of Prop.~\ref{prop:outgoing-edges-from-char-all-same-color}, $x$
precedes the edge $e_1= (v,z)$ that is labeled by $c_1^{-}$, as $c_1^{+}$
belongs to the label sequence of $x$, but  $c_1^{-}$ does not.
Two cases are possible. If $c_1^{+}$  belongs to the label sequence for $y$,
since
both $s_{1}$ and $s_2$ do not have character $c_1$ but are descendant of $y$,
they must also be two descendants of $z$ and a fortiori of $x$, a contradiction
of the initial assumption. Thus, assume that  $c_1^{+}$  does not belong to the
label sequence of $y$.

Thus there exists an edge $(u',v')$ labeled $c_1^{+}$ that occurs along the path
from node $y$ to $u$.
In the graph $G_{RB}(v')$, $c_1$ is connected to $s_1$ and $s_2$ with red edges,
for otherwise the same property would not hold for $G_{RB}(v')$.
Since we have asked for a vertex $x$ that is the closest to the root for
which this Lemma does not hold, the Lemma does hold for all vertices on the path
from the root to $u$.
Consequently, $s_{1}$ and $s_{2}$ are both descendants of $u$, contradicting the
hypothesis that $y$ is the lca of $s_{1}$ and $s_{2}$.
\qed\end{proof}

{\bf Proof of Lemma \ref{lemma:2-tcs} }
\begin{proof}
Consider the tree $T$ realizing $M_e$ that satisfies
Lemma~\ref{lemma:2-tcs-part1}.
It is immediate to modify $T$ in such a way that it has properties 1-7 of definition of a standard tree.
Let $e=(u,x)$ be an edge labeled by $c^{+}$ or $c^{-}$, and let $S(x)$, $C(x)$
be respectively the species and the characters of $T(x)$.
Notice that an immediate consequence of Lemma~\ref{lemma:2-tcs-part1} is that
the set of species that are in the same
connected component $\mathcal{C}$ of graph $G_{RB}(u)$ is either contained in
$T(x)$ or  disjoint from $T(x)$.
Therefore $T(x)$ is the union of some connected components of $\grb(u)$.
It is easy to show that properties 8-9 of a standard tree holds for $T$.
\qed\end{proof}

{\bf Proof of Lemma \ref{main-reduction0}}
\begin{proof}
By definition of standard tree and of depth-first visit, each negative character
$c^{-}$ is preceded in $V$ by $c^{+}$. Therefore we only have to prove that all
all single character reductions of $V$ are possible and the graph
reduced by $V$ has no edge.

Let us prove the lemma by induction on the number $n$ of nodes of $T$.
Clearly the lemma holds when $n=1$, that is the tree consists of a single leaf,
therefore assume that the lemma holds for all tree with at most $n$ nodes.

If \grb is disconnected, then it is easy to
 show  that visiting separately and successively each subtree of  a standard tree $T$ is equivalent to
visiting all of $T$, from which the lemma easily follows.

Assume now that \grb is connected, therefore the root $r$ of $T$ has only a
child $x$.
Let $e=(r,x)$ be the only edge of $T$ incident on $r$ and let $c$ be the
character labeling $e$.
Since $x$ is the only child of $r$, $V$ consists of the edge $e$ followed by the
depth-first visit of $T(x)$.

By definition of standard tree (point 6), $\grb(x)$ is obtained from \grb by
realizing character $c$.
Consequently $T(x)$ solves the graph $\grb(x)$ and, by inductive hypothesis, the
lemma holds for $T(x)$.
Moreover $T(x)$ contains no black edges incident on $c$, as they all have been
removed when realizing $c$.
All other edges of $\grb$ are removed by realizing the characters in the
depth-first visit of $T(x)$, therefore no edge remains from \grb after
realizing $V$.
\qed\end{proof}

{\bf Proof of Theorem  of \ref{equivalence-0}}
\begin{proof}

Let   $M$ be the completion of matrix \me\ obtained from a successful c-reduction
of the red-black graph for $M_e$.
In the following we show that $M$ has no forbidden matrix. This fact will prove
that $M$ admits a pp tree.
Let $G_R$ be the red-black obtained after the realization of the characters of
the successful  c-reduction.
 Assume to the contrary that $M$ has two characters $c, c_1$ that induce a
forbidden matrix $F$, and let $s_1, s_2, s_3$ be the species of $M$  having the
configuration $(1,1)$, $(1,0)$ and $(0,1)$ in $F$, respectively.

We must consider the following cases.

Case 1. Assume that the forbidden matrix is induced by   two negated characters.
This fact implies that $G_R$ will have an induced $\Sigma$-graph, thus
contradicting the fact that $G_R$ is edgeless.

Case 2. Assume that the forbidden matrix is induced by two positive characters.

 Then  $c, c_1$ must be in the same connected component of the  red-black graph
before their realization, as species  $s_1$ is connected to both characters (we
do not know if $s_1$ is connected  by a black or red edge).
 Now, the realization of $c$ produces the red edge $(c,s_3)$,  since
$M_e[s_3,c] = ?$. Then $M[s_3,c]=1$ in the completion  $M$, a contradiction
with the assumption.

 Case 3. Assume that the forbidden matrix is induced by a positive and negated
character.

Assume that $c$ is the negated character.  Since $c$ and $c_1$ share a species
in the forbidden matrix $F$, it means that $c$ and $c_1$ are in the same
connected component of the red-black graph  when $c_1$  and $c$ are realized in
the graph. Since $(0,1)$ is given in the matrix $F$ in row $s_2$, by definition
of realization of $c_1$, it must be that  $M[s_2, c_1]=1$ and
$M[s_2,\bar{c_1}]= 1$ as $M_e[s_2,c_1]=0$. But this  is a contradiction.

\qed\end{proof}

 {\bf Proof of Lemma \ref{clique}}
  \begin{proof}
  %------da rivedere -----

 Let $a,b$ be an arbitrary  pair of elements that are in set $C_M$.
 %We show that $(a,b)$ is an edge of the adjacency graph.
  Since the red-black is connected, there exists a path $\pi$  connecting the two vertices $a$ and $b$. Then by induction on the number $k$ of  internal vertices of  the smallest  path that connects $a$ to $b$ we show that $(a,b)$ is an edge of the adjacency graph $G$.
Assume first that $k=1$.   Then assume that
  $(a,c)$, $(c,b)$ are edges of the  shortest path  $\pi$ connecting the two vertices $a$ and $b$.

Assume on  the contrary that there is no edge $(a,b)$ in the adjacency graph.
%We will obtain a contradiction, since  we can show the existence of  a species $s$ that is not connected to both $a, c$ in %the red-black graph.
The following cases must be considered. Assume that $c$ is not comparable with both $a$ and $b$. Since $c$ adjacent with $a$ it follows that there exists a species $s$ such that $s$ contains $a$ but not $c$ and similarly there exists a species $s'$ such that $s'$ contains $c$, but not $a$. But this fact would imply that $(a,c)$ is an edge of the conflict graph which is a contradiction. For the same reason $c$ must be comparable with $b$.
Assume that $c$ is comparable with $b$ and $a$ in $P$.
Since $a, b \in C_M$, it  holds that $c < b$ and $c < a$. But this fact would imply that $a, b$ are adjacent in graph $G$ as they share a common species.
Then assume that $k =n >1$. Clearly,      $G$  must have edge $(a, c)$ where $c$ is the vertex adjacent to $a$ on the path  having $k$ internal vertices. Since the   path connecting $c$ and $b$  has   $k-1$ internal vertices, by induction it holds that  also   edge $(c,b)$ is in graph $G$.  Consequently, there exists a   path with $1$ vertex  in the red-black graph connecting  $a$ and $b$. By induction, we show that $(a,b)$ must be an edge of the graph $G$, as required.

 \qed

  \end{proof}

  {\bf Proof of Lemma \ref{consecutive-p}}
\begin{proof}
Take two characters $a$ and $b$  that are consecutive in the graph \grbd

Then two cases are possible.
Case 1:  assume that  $a< b$ or $b < a$. Then the realization of $a$ and $b$ produces  a unique red graph that is connected, otherwise there exists a species to which $a$ is connected but $b$ is not  and vice-versa, which contradicts the containment relation between the two characters. This fact proves the lemma.

Case 2: if $a$ and $b$ are not comparable, then   there exists a species $s$  such that $a$ connected  to $s$ but $b$ is not connected to $s$, and vice-versa there exists a species $s'$ such that $b$ is connected to $s'$, but $a$ is not.
Then there does not exist a specie $x$ of  graph \grb  that is not connected to both   $a$ and $b$.  In fact,  if  on the contrary, such a species exists, then the pair $(a,b)$ induces the four gametes  and thus $(a,b)$ is an edge of  the conflict graph, contradicting the fact that it is empty. Consequently, if species $x$ does not exist,  then the realization of $a$ and $b$ induces two red disjoint components thus proving the lemma.
\qed

\end{proof}

\begin{theorem}
\label{pol_empty-1}
Let $(M, E^*)$ be a binary matrix that has an empty conflict graph. Then there exists a polynomial time algorithm to build the p-pp tree for $M$ and if $E^*$ is empty then $M$ admits a solution.
\end{theorem}

%{\bf Proof of  Theorem \ref{pol_empty}}
\begin{proof}
Clearly, the algorithm for an empty conflict graph described in the paper works in polynomial time in the size of the input matrix $M$. In the following we show that the algorithm outputs a sequence $r$ that is a successful c-reduction of  graph \grbd. %(clearly  if $E^*$ is empty,  such a reduction always exists).
In the following we assume that  a maximal character $c$ can be realized before a character $c'$ such that  $c' < c$ in a successful c-reduction.
In particular, in the case that of the P-PP problem a solution always exists. In fact  at each iteration all characters may be realized and in the following we show that no red sigma-graph is induced at any iteration of the algorithm, till the red-black graph is edgeless.
By Lemma \ref{clique}, the set $C_M$  of $P$ forms a clique. Then by Lemma \ref{consecutive-p}, these characters can be realized in any arbitrary order without inducing  a $\Sigma$-graph  in the red-black graph.  More precisely, by Lemma \ref{consecutive-p}, after the realization of set $C_M$,  each character in   $C_M$ is connected to  red-edges and the red components induced by  a single character in $C_M$ are disjoint .
Observe that at a second iteration of step 1,  a maximal element  $c$ is    contained in a  maximal element  $c_m$ of the first iteration of step 1.

In the following we show that after the realization of elements in $C_M$,
 the red-black graph consists of disjoint components and each maximal element  $c$ that is in the poset $P$ after the removal of $C_M$ is in one of  such disjoint connected components and   it is a universal characters or the component has a unique maximal character in $C_M$ that is free. In the first case, no red-edge is induced by the realization of  $c$, while in the second case, the red-edges of the component induced  by an element in $C_M$ are removed. Consequently, it follows that by the application of   Lemma \ref{consecutive-p} and Lemma \ref{clique}, no red sigma-graph is induced in the new iteration of the algorithm.
%Moreover, we can show that each character in $C_m$ is free in one of the two components.

Let $c$ be an element such that $c < c_m$.
%By definition of  relation $<$, since $c < c_m$, it holds that
%character $c$ is contained in the connected component of another  element $c_n$ in $C_M$ or $c$ is in  a component that is disjoint from all the others.
 Assume that  $S_1 = S(c_n) / S(c_m)$ and $S_2= S(c_m) / S(c_n)$. Clearly $S_1$ and $S_2$  are disjoint sets and since $S(c) \subseteq  S(c_m)$, it holds that $S(c)/ S(c_n) \subset S_2$. Now, we have three cases: (1)  there exists a species $s \in S_2$ such that $s$ is not in $S(c)$ and a species $s_1 \in S_2$ such that $s_1$ is   in $S(c)$, or (2) all species in $S_2$ are not in $S(c)$  or (3)  all species in $S_2$  are  in $S(c)$ .  Assume that case 1) holds. Then,  if there exists a species $s'$ that is of $c_n$ and $c_m$ and $s'$ has character $c$,  it holds that $c$ and $c_n$ form an edge  of the conflict graph, thus this case is not possible. Consequently, if case 1) holds it means that  $c$ is in the component consisting only of species in $S_2$ and $c_n$ is free in such a component.
  Assume that case (2) holds. Then $c$ is connected to a species that is of $c_n$ and $c_m$. It follows that $c$ is universal for the component consisting of the species of $c_n$ and $c_m$ to which $c$ is connected.  Assume now that case 3) holds. Again this case shows that $c$ is universal for the connected component to which $c$ belongs.

  % By combining the three cases, it holds that either  maximal elements are free in a connected component or otherwise it is in a component where a maximal element in poset $P / C_M$ is universal for the %component.  In the first case it follows that we can iterate  the  application of Lemmas \ref{clique} and \ref{consecutive-p} to the   single connected  component  of the red-black graph where the maximal %element is free, since all red edges are removed from the component. Otherwise, the realization of the maximal element which is universal for the component allows the application again of
%Lemmas \ref{clique} and \ref{consecutive-p} for the new component. In other words we can realize the maximal elements of $P /C_M$.

% then $c_n$ is incident to  red edges connecting species to which $c_m$ was connected before the realization of $C_M$. Thus $c$ is connected to some of these species. It follows that the red-black graph consists of disjoint components after the realization of elements in $C_M$. Each of this component has an element $c_m$ that is free, that is all incident edges to such character are deleted from the red-black graph.
%Thus we re  the  application of Lemmas \ref{clique} and \ref{consecutive-p} to the maximal elements in the poset $P$ that has been updated and thus to single connected  components  of the red-black %graph.
\qed
\end{proof}

We shows that    the  CP-PP  problem is equivalent to a case of the General  Character Compatibility~\cite{DBLP:conf/isbra/ManuchPG11}  (GCC) whose computational complexity is still open and similarly we show the connection of the  restricted P-PP problem with the GCC one.

 In the General  Character Compatibility problem, the input is a set $S$ of species. Observe that a set of species $S$ is defined by means of a set of {\em generalized characters}.
A generalized character $\hat{\alpha}$ consists of a pair $(\alpha, T_{\alpha})$
where $\alpha: S \rightarrow 2^{Q_{\alpha}}$ ($Q_{\alpha}$ is the set of states
for $\hat{\alpha}$), and $T_{\alpha}= (V (T_{\alpha}),E)$ is a rooted character tree with nodes bijectively labeled by the elements of $Q_{\alpha}$.

%For example, we may have two generalized characters $\hat{\alpha} = ( \alpha(s_1) = 2, \alpha(s_2) = 0, T_{\alpha} =  0 \rightarrow 1 \rightarrow 2)$.

Then a the solution of  the GCC problem is a rooted tree $T = (V, E)$   and a {\em state-assignment} function
$c$ that assigns to each vertex $v \in V(T)$ and generalized character $\hat{\alpha}$  a state, in such a way that $c(v, \hat{\alpha}) \in Q_{\alpha}$ such that the following fact holds.

\begin{enumerate}

\item For each species $s \in  S $ there is a vertex  $v_s$ of tree $T$  such
that $c(v_s, \hat{\alpha}) \in \alpha (s)$ for each $\hat{\alpha}$.

\item For every $\hat{\alpha}  \in C$  and state $i \in Q_{\alpha}$, the set of vertices of $T$ where $\hat{\alpha}$  assumes value $i$  is a connected
component of $T$.

\item  For every $\hat{\alpha}  \in C$, the tree $T({\alpha})$ is an induced subtree of $T_{\alpha}$, where  $T({\alpha})$ is  obtained from $T$  by labeling the nodes of  $T$ only with their  $\alpha$-states
(as chosen by function $c$), and then contracting edges having the same $\alpha$-state at their
endpoints.
\end{enumerate}

In the following we show that the P-PP problem is equivalent to the GCC problem where $ {\alpha}(s) \in \{\{1\}, \{0,2\}\}$ for each species $s$ and moreover $T_{\alpha} =  0 \rightarrow 1 \rightarrow 2$.

%Let  $M$ be a generic instance of the  P-PP problem. Then we build an instance of the GCC problem as follows:
%we assign   a generalized character $\hat{\alpha}_j$ to each character $j$ of
%$M$ so that  if $M[i,j] = 0$, then $ \hat{\alpha}_j(s_i)  \in \{0,2\} $,
%otherwise $\hat{\alpha}_j(s_i)  \in \{1\} $.
%Moreover,  $T_{\alpha_j} =  0 \rightarrow 1 \rightarrow 2$ or $T_{\alpha} =  0 \rightarrow 1 $, for each character $\hat{\alpha}_j$.

We can show that state $2$ in the GCC problem is equivalent to the fact that  a character is  persistent that is, it is in state $0$ (lost)  but it has  been in state  $1$ before (gain).
Given an instance of the  GCC problem where for each generalized character $\hat{\alpha}$ it holds that $ {\alpha}(s) \in \{\{1\}, \{0,2\}\}$ for each species $s$ and $T_{\alpha} =  0 \rightarrow 1 \rightarrow 2$, we can associate a matrix $M$ that is an instance of the P-PP problem as follows: we assign a  column $c_{j}$ in $M$ to each generalized character  $\hat{{\alpha}_j}$ and a row for each species $s \in S$, where
$M[s, c_{j}] = 0$ if and only if  $ \hat{{\alpha}_j} (s )  \in \{0,2\} $, otherwise  $M[s, c_{j}] = 1$.
  Given a solution  $T_{pp}$ of the PP-problem,  it is easy to verify that the underlying unlabeled tree $T$ together with the function  $c$ such that for  each vertex $v$  and vector $l_v$, it holds that $l_v[i]  = c(v,\hat{\alpha}_i)$, if $c_i$ is in state $0$ or $1$ for the first time, otherwise if   $l_v[i] =0$ for the second time, then $c(v,\hat{\alpha}_i) =2$ is a solution of the associated GCC problem instance.
Consequently,  GCC reduces to the P-PP-problem.

The vice versa, that is the P-PP-problem reduces to GCC, is easily proved as follows. For each species $s$ of an input matrix $M$ of the PP-problem and character $c_j$, we define  the generalized character $\hat{{\alpha}_j}$, where $\hat{{\alpha}_j}(s )  \in \{0,2\} $ if and only if $M[s,c_j]=0$ otherwise $\hat{{\alpha}_j}(s )  \in \{1\} $ if and only if  $M[s, c_j] = 1$. Moreover $T_{{\alpha}_j} =   0 \rightarrow 1 \rightarrow 2$ for each $\hat{{\alpha}_j}$.

  In fact we can show that given a  tree  $T$  which is solution of the GCC problem,  then
  we can  build a    tree $T_{pp}$ obtained  by   assigning to each vertex $v$ of the tree $T$  a vector $l_v$ such that $l_v[i] =c(v, \hat{\alpha}_i)=j $,  when  $j \in \{0,1\}$ otherwise if $j = 2$, then $l_v[i] =0$.

Now,  $T_{pp}$   is a solution of the initial  P-PP problem instance.  In fact the four properties of Definition~\ref{def:persistent-perfect-phylogeny}  holds for $T_{pp}$.
In particular condition 4 of such definition is proved as follows.
Notice  that,  by  property 3 of tree $T$, the  tree $T({\alpha})$ represents the change in states of each character $\hat{\alpha}$. Since  $T_{\alpha} =  0 \rightarrow 1 \rightarrow 2$,  the state  of a character changes  from $0$ to $1$ and again to $0$ (which we represent with state $2$) only if persistent.
 A direct consequence of property 2 of tree $T$,  is that in tree $T$ we can label only one edge $(u,v)$ of the tree  with  a character $c_{j}^+$ or $c_{j}^-$, more precisely, when     then the state of the generalized character $\hat{{\alpha}_j}$ changes from $0$ to $1$ in $l_u$ and $l_v$, and respectively, from $1$ to $2$ in $l_u$ and $_v$  in the tree $T$. In fact, the set of vertices with a common state $c(v, \hat{\alpha}_j)=k$   for a character $\hat{\alpha}_j$  must be a connected component. Now, when $k\in \{1\}$  (or $k \in \{2\}$),  the connected component induced by    vertices  $v$ labeled $k$ by the $c$ function represents the substree of tree $T_{pp}$ rooted in the end of the edge labeled $c_{j}^+$ (or $c_{j}^-$, respectively).

Using the above result,  assuming that some characters may have only state $0$ in a given species,  it is easy to show that the CP-PP problem is equivalent to the GCC in the case that $ {\alpha}(s) \in \{\{1\}, \{0\}, \{0,2\}\}$ for each species $s$ and $T_{\alpha} =  0 \rightarrow 1 \rightarrow 2$.

\end{document}